\documentclass[a4paper,english,numberwithinsect]{eurocg25}

\usepackage{physics}
\usepackage{mathtools}
\usepackage{bm}
\usepackage{subcaption}
\usepackage{complexity}
\usepackage{amsmath,amsthm}
\usepackage{amssymb}
\usepackage{thmtools}
\usepackage{etoolbox}
\usepackage{amsthm}
\usepackage[capitalize,nameinlink]{cleveref}
\usepackage{soul}
\usepackage{xspace}

\definecolor{defblue}{RGB}{84,39,143}
\definecolor{lipicsblue}{rgb}{0.08235294118,0.3098039216,0.537254902}
\definecolor{tblue}{RGB}{69,99,169}
\definecolor{tgreen}{RGB}{92, 136, 43}
\definecolor{tred}{RGB}{159,29,39}
\newcommand{\blue}{\textbf{\color{tblue}blue}\xspace}
\newcommand{\green}{\textbf{\color{tgreen}green}\xspace}
\newcommand{\red}{\textbf{\color{tred}red}\xspace}
\let\emph\relax
\DeclareTextFontCommand{\emph}{\color{defblue}\em}
\hypersetup{colorlinks=true,
    linkcolor=defblue,
    anchorcolor=defblue,
    citecolor=defblue,
    filecolor=defblue,
    menucolor=defblue,
    urlcolor=defblue,
    bookmarksopen=true,
    bookmarksopenlevel=2,
    bookmarksnumbered=true,
    plainpages=false
    }

\DeclareMathOperator{\ext}{ext}

\DeclarePairedDelimiterX\set[1]\lbrace\rbrace{#1}

\title{Saturated Drawings of Geometric Thickness \texorpdfstring{$\bm{k}$}{k} \footnote{This research was initiated at GGWeek 2024 in Trier. We would like to thank the organizers and participants of the workshop for the friendly and supportive environment and the fruitful discussions. }}
\titlerunning{Saturated Drawings of Geometric Thickness \texorpdfstring{$\bm{k}$}{k}}

\ArticleNo{}

\author[1]{Patricia~Bachmann}
\author[2]{Anna~Brötzner\footnote{supported by grant 2021-03810 from the Swedish Research Council (Vetenskapsr\r{a}det).}}
\author[3]{Miriam~Goetze\footnote{funded by the Deutsche Forschungsgemeinschaft (DFG, German Research Foundation) -- 520723789}}
\author[4]{Philipp~Kindermann}
\author[1]{Matthias~Pfretzschner\footnote{funded by the Deutsche Forschungsgemeinschaft (DFG, German Research Foundation) -- 541433306}}
\author[5]{Soeren~Terziadis\footnote{funded by the NWO Gravitation project NETWORKS under grant no. 024.002.003.}}

\affil[1]{Chair of Theoretical Computer Science, University of Passau, Germany, \\ \texttt{bachmanp@fim.uni-passau.de}, \texttt{pfretzschner@fim.uni-passau.de}}
\affil[2]{Department of Computer Science and Media Technology, Malm\"o University, Sweden,
  \texttt{anna.brotzner@mau.se}}
\affil[3]{Institute of Theoretical Informatics, Karlsruhe Institute of Technology, Germany, \texttt{miriam.goetze@kit.edu}}
\affil[4]{Algorithms Group, Trier University, Germany, \texttt{kindermann@uni-trier.de}}
\affil[5]{Algorithms Group, TU Eindhoven, The Netherlands, \texttt{s.d.terziadis@tue.nl}}

\authorrunning{P. Bachmann, A. Brötzner, M. Goetze, P. Kindermann, M. Pfretzschner, S. Terziadis}

\makeatletter
\def\renewtheorem#1{%
  \expandafter\let\csname#1\endcsname\relax
  \expandafter\let\csname c@#1\endcsname\relax
  \gdef\renewtheorem@envname{#1}
  \renewtheorem@secpar
}
\def\renewtheorem@secpar{\@ifnextchar[{\renewtheorem@numberedlike}{\renewtheorem@nonumberedlike}}
\def\renewtheorem@numberedlike[#1]#2{\newtheorem{\renewtheorem@envname}[#1]{#2}}
\def\renewtheorem@nonumberedlike#1{  
\def\renewtheorem@caption{#1}
\edef\renewtheorem@nowithin{\noexpand\newtheorem{\renewtheorem@envname}{\renewtheorem@caption}}
\renewtheorem@thirdpar
}
\def\renewtheorem@thirdpar{\@ifnextchar[{\renewtheorem@within}{\renewtheorem@nowithin}}
\def\renewtheorem@within[#1]{\renewtheorem@nowithin[#1]}
\makeatother

\newtheorem{proposition}[theorem]{Proposition}
\newtheorem{observation}[theorem]{Observation}
\renewtheorem{obs}[theorem]{Observation}
\Crefname{observation}{Observation}{Observations}
\Crefname{obs}{Observation}{Observations}
\renewtheorem{lemma}[theorem]{Lemma}
\crefname{lemma}{Lemma}{Lemmas}
\Crefname{lemma}{Lemma}{Lemmas}
\crefname{corollary}{Corollary}{Corollaries}
\Crefname{corollary}{Corollary}{Corollaries}

\usepackage{apptools}
\newcommand{\restateref}[1]{\IfAppendix{\hyperref[#1]{$\star$}}{\hyperref[#1*]{$\star$}}}

\begin{document}

\maketitle

\begin{abstract}
We investigate saturated geometric drawings of graphs with geometric thickness~$k$, where no edge can be added without increasing $k$. 
We establish lower and upper bounds on the number of edges in such drawings if the vertices lie in convex position. We also study the more restricted version where edges are precolored, and for $k=2$ the case for vertices in non-convex position.
\end{abstract}

\section{Introduction}
The \emph{geometric thickness} $\bar{\theta}(G)$ of a graph $G$ is the minimum number~$k$ such that there exists a straight-line drawing~$\Gamma$ of~$G$ and a $k$-edge-coloring~$\varphi\colon E(G) \to \{1,\dots,k\}$ that has no monochromatic crossings, see \cref{fig:non-convex-thickness_example} for an example.
\begin{figure}
    \centering
    \includegraphics[page=2]{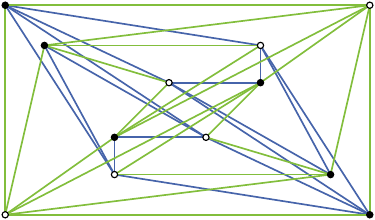}
    \caption{Taken from \cite[Fig. 2]{dillencourt2004geometric}. A drawing of the non-planar graph~$K_{6,6}$ witnessing $\bar{\theta}(K_{6,6}) = 2$.}
    \label{fig:non-convex-thickness_example}
\end{figure}
We also write $E_i \subseteq E(G)$ to denote all edges of color $i$.
We call~$\Gamma$ a \emph{$\Theta^k$-drawing} (with thickness $k$).
When the coloring~$\varphi$ of $\Gamma$ is given, we say that~$\Gamma$ is \emph{precolored} (we always assume that in a given coloring, there are no monochromatic crossings).
If the vertices in $\Gamma$ are in convex position, $\Gamma$ is \emph{convex}.
Connecting all vertices of the outer face of a $\Theta^k$-drawing with edges along the convex hull yields a cycle.
We call this cycle the \emph{outer cycle} of~$\Gamma$, and an edge~$e$ on this cycle an \emph{outer edge} of~$\Gamma$. 
All other edges are \emph{inner edges}. 
Note that not all edges of the outer cycle are necessarily contained in $\Gamma$.

We call a $\Theta^k$-drawing~$\Gamma$ of a graph~$G$ \emph{saturated} if there are no two vertices $u,v \in V(G)$ with~$uv \notin E(G)$ such that the drawing~$\Gamma'=\Gamma+uv$ with the edge~$uv$ drawn as a straight line is also a $\Theta^k$-drawing.
If $\Gamma$ is precolored, we require that $\Gamma'$ uses the same coloring, i.e., only the color of $uv$ may vary.
If~$\Gamma$ has the minimum (maximum) number of edges among all saturated $\Theta^k$-drawing on the same number of vertices, it is \emph{min-saturated} (\emph{max-saturated}).
We assume that vertices lie in general position, i.e., there are no three vertices on a line.

Max- and min-saturation have similarly been defined for graph classes instead of drawings.
There is a rich history of results analyzing max-saturated graphs (Tur\'an type results, following seminal work by Tur\'an~\cite{turan2007extremalaufgabe}).
It is widely known that max-saturated planar graphs contain $3n-6$~edges, and bounds have been proven for several beyond planar graph classes.
For example, $1$-planar and $2$-planar max-saturated graphs have $4n-8$ and $5n-10$ edges, respectively \cite{pach1997graphs}, while general $k$-planar max-saturated graphs are only known to have at most $3.81\sqrt{k}n$ edges~\cite{ackerman2019topologicalGraphs}.
Similar results have recently been shown for min-$k$-planar graphs~\cite{binucci2024minkplanarDrawings}.

The study of min-saturated graphs builds on the work of Erd\H{o}s, Hajnal, and Moon~\cite{erdos1964problem}, who characterize min-saturated $K_k$-free graphs.
A survey~\cite{faudree2011survey} with a recent second edition provides an overview of results in this direction. While min-saturated planar graphs also contain $3n-6$ edges, min-saturated 1-planar graphs only have at most $\frac{45}{17}n+O(1)$ edges.
Chaplick et al.~\cite{chaplick2024edge} recently investigated the number of edges in min-saturated (not necessarily straight-line) $k$-planar drawings
under a variety of drawing restrictions. 

Graphs of geometric thickness $k$ form a relevant beyond-planar graph class.
The concept was first introduced by Kainen~\cite{kainen1973thickness} (who used the term linear thickness) and later investigated by Dillencourt, Eppstein, and Hirschberg~\cite{dillencourt2004geometric}, who considered the geometric thickness of complete and complete bipartite graphs.
Checking whether a graph has geometric thickness at most $k$ has been shown to be \NP-hard~\cite{durocher2016thickness} even for $k\leq 2$ and for multigraphs it is 
$\exists\mathbb{R}$-complete~\cite{forster2024geometric} for $k\leq 30$.
In fact, a graph~$G$ has stack number at most~$k$ if and only if it admits a convex $\Theta^k$-drawing. 
That is, in the convex setting, we investigate the min-saturation of graphs with stack number at most~$k$.

We provide upper and lower bounds on the number of edges in min-saturated $\Theta^k$-drawings in the precolored and non-precolored, as well as in the convex and non-convex setting.
After presenting upper bounds for convex precolored and non-precolored $\Theta^k$-drawings in \cref{subsec:general_bounds}, we give lower bounds for $\Theta^3$-drawings (applying to the precolored and non-precolored setting) in \cref{subsec:small_k}.
In \cref{sec:non-convex-free}, we present a lower bound for non-convex non-precolored $\Theta^2$-drawings and conclude in \cref{sec:conclusion}. 
Results marked with $(\star)$ are proved in the appendix.

\section{Convex Drawings}\label{sec:convex-free}
Each color class of a convex $\Theta^k$-drawing induces an outerplane graph~$H$.
For~$\ell \geq 3$, we call an outerplane graph~$H$ an \emph{inner $\ell$-angulation} if every inner face has size~$\ell$ and the outer face is a simple cycle. 
Inner $3$-angulations and inner $4$-angulations are called \emph{inner triangulations} and \emph{inner quadrangulations}, respectively.
Double-counting the edge-face-incidences shows that every inner $\ell$-angulation with $n$~vertices and $f$~faces contains $\frac{1}{2}(n+\ell(f-1))$~edges.
Now Euler's formula implies:
\begin{observation}
\label{obs:edges_in_inner_triang_inner_quadrang}
    For~$\ell \geq 3$, every inner $\ell$-angulation of a graph on $n \geq \ell$ vertices contains $\frac{n-\ell}{\ell-2}$~inner edges.
\end{observation}

\subsection{Bounds for Saturated Convex \texorpdfstring{$\bm{\Theta^k}$}{Theta-k}-Drawings}
\label{subsec:general_bounds}

Note that each color class of a convex~$\Theta^k$-drawing of an $n$-vertex graph is a subgraph of some inner triangulation.
That is, we can cover the edges of the~$\Theta^k$-drawing with $k$~inner triangulations, any two of which only share the outer cycle.
Now, \cref{obs:edges_in_inner_triang_inner_quadrang} yields the following upper bound on the number of edges.

\begin{proposition}[{\cite[Theorem 3.3]{bernhart1979book}}]
\label{obs:general_upper_bound_convex}
    Every convex $\Theta^k$-drawing of a graph~$G$ on $n \geq 3$~vertices contains at most $n+k(n-3)$~edges. 
\end{proposition}

We now construct a precolored saturated drawing with a smaller number of edges than implied by \cref{obs:general_upper_bound_convex}, thereby obtaining a smaller upper bound for precolored min-saturated drawings. 
We say that some diagonals~$M$ of a convex $\Theta^k$-drawing of a graph on $n$~vertices form a \emph{nice matching} if these diagonals together with the outer cycle form an outerplane graph~$H$ whose dual is a path (ignoring the outer face) and where the faces corresponding to the beginning and end of the path are faces of size~$3$ or~$4$, and all other faces have size~$4$, see \cref{fig:convex_k_min_upper_bound-1}.
\begin{figure}
    \centering
    \subcaptionbox{\label{fig:convex_k_min_upper_bound-1}}{\includegraphics[page=1]{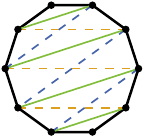}}
    \hfil
    \subcaptionbox{\label{fig:convex_k_min_upper_bound-2}}{\includegraphics[page=2]{figures/convex_k_min_upper_bound.pdf}}
    \hfil
    \subcaptionbox{\label{fig:convex_k_min_upper_bound-3}}{\includegraphics[page=4]{figures/convex_k_min_upper_bound.pdf}}
    \hfil
    \subcaptionbox{\label{fig:convex_k_min_upper_bound-4}}{\includegraphics[page=5]{figures/convex_k_min_upper_bound.pdf}}
    \caption{(\subref{fig:convex_k_min_upper_bound-1}) A nice matching~$M$ (green) and diagonals that could be added to a color class containing~$M$ (dashed). (\subref{fig:convex_k_min_upper_bound-2})
    The  two nice matchings~$L(M)$ (left) and~$R(M)$ (right). (\subref{fig:convex_k_min_upper_bound-3}) A precolored $\Theta^3$-zigzag~$\Gamma$. (\subref{fig:convex_k_min_upper_bound-4})   
    Recoloring yields a $\Theta^3$-drawing which contains $\Gamma$ and two more edges (dashed).}
    \label{fig:convex_k_min_upper_bound}
\end{figure}
If all these diagonals belong to the same color class $E_i$, then $E_i$ can only be extended by adding missing diagonals within the faces of~$H$.
The missing diagonals may again be decomposed into two nice matchings, which we call the \emph{left} and \emph{right tilt} of~$M$, denoted by~$L(M)$ and~$R(M)$, respectively, see \cref{fig:convex_k_min_upper_bound-2}.
In particular, we have~$R(L(M)) = M$ and~$L(R(M)) = M$.
We can now construct a saturated $\Theta^k$-drawing~$\Gamma$ of a graph~$G$
on $n$~vertices with an edge-coloring~$\varphi\colon E(G) \to \set{1, \dots, k}$ such that the following holds (see \cref{fig:convex_k_min_upper_bound-3} for an example):
\begin{itemize}
    \item The outer cycle is part of~$\Gamma$
    \item The inner edges of $E_1$ correspond to two nice matchings~$M_1$ and~$L(M_1)$
    \item The inner edges of $E_{i}$ form a nice matching~$M_i = R(M_{i-1})$, for $i=2,\dots,k-1$ 
    \item The inner edges of $E_k$ correspond to two nice matchings~$M_k = R(M_{k-1})$ and~$R(M_k)$
\end{itemize}
We call the obtained precolored drawing a \emph{precolored $\Theta^k$-zigzag}~$\Gamma$ on $n$~vertices.
Here, no $E_i$ can be extended as all the edges that could be added to~$E_i$ are part of some $E_j$ with $j \not= i$.
That is, the $\Theta^k$-zigzag~$\Gamma$ is a precolored saturated drawing.
For $k \leq \frac{n}{2}$ we have $E_i \cap E_j = \emptyset$, i.e., disjoint edge sets and~$\Gamma$ is well-defined.

\begin{proposition}
\label{prop:upper_bound_precoled_min_sat}
    Every min-saturated convex precolored $\Theta^k$-drawing of a graph~$G$ on $n \geq 5$~vertices (with $k \leq \frac{n}{2}$) contains at most $\frac{1}{2}(k+4)(n-2)$~edges.
\end{proposition}
\begin{proof}
Consider the precolored $\Theta^k$-zigzag on $n$~vertices.
By \cref{obs:edges_in_inner_triang_inner_quadrang}, $E_1$ and~$E_k$ contain at most $n-3$ inner edges respectively.
Every nice matching together with the outer cycle is an inner quadrangulation except for at most two faces of complexity 3.
A similar argument as in \cref{obs:edges_in_inner_triang_inner_quadrang} shows that every nice matching contains at most $\frac{1}{2}(n-2)$~edges.
Summing up the number of edges of the outer cycle ($n$ edges), the inner edges of~$E_1$ and~$E_k$ ($2(n-3)$), and the edges of the $k-2$ nice matchings~$E_i$
($\frac{1}{2}(n-2)$ each) yields the desired bound.
\end{proof}
Yet, this upper bound does not yield an upper bound for non-precolored drawings.
Indeed, a $\Theta$-zigzag (without the edge-coloring) is not necessarily saturated, cf. \cref{fig:convex_k_min_upper_bound-3} and \cref{fig:convex_k_min_upper_bound-4}.

Recall that every color class of a~$\Theta^k$-drawing together with the outer cycle forms an outerplane graph. 
For max-saturated precolored drawings, the inner faces of these outerplane graphs cannot have arbitrarily large size:

\begin{restatable}[\restateref{obs:face_complexity}]{lemma}{facecomplexity}
\label{obs:face_complexity}
    If $\Gamma$ is a saturated precolored convex $\Theta^k$-drawing, then each color class of~$\Gamma$ together with the outer cycle forms an outerplane drawing where each inner face has size at most~$2k-1$.
\end{restatable}

Thus, by \cref{obs:face_complexity}, we can cover the edges of a saturated precolored $\Theta^k$-drawing with $k$~outerplane graphs, each of which contains an inner $(2k-1)$-angulation that contains the edges of the outer cycle.
An application of \cref{obs:edges_in_inner_triang_inner_quadrang} yields the following.

\begin{theorem}
\label{obs:convex_lower_bound_k}
    Every min-saturated convex precolored $\Theta^k$-drawing of a graph on $n \geq 2k-1$ vertices contains at least $\frac{k(n-2k+1)}{2k-3}+n$~edges.
\end{theorem}

Note that, since the upper bounds implied by \cref{obs:general_upper_bound_convex} and  \cref{prop:upper_bound_precoled_min_sat} and the lower bound of \cref{obs:convex_lower_bound_k} coincide for $\Theta^2$-drawings, we obtain the following.

\begin{corollary}
\label[corollary]{obs:lower_bound_2}
    Every saturated (precolored) convex $\Theta^2$-drawing~$\Gamma$ of a graph~$G$ on $n \geq 3$ vertices contains exactly $3n-6$ edges.
\end{corollary}

Note that the number of edges of saturated convex $\Theta^2$-drawings only depends on the number of vertices, that is, min- and max-saturated $\Theta^2$-drawings coincide.
This is different from other results related to saturation problems.
For example, there are saturated $2$-planar drawings of graphs on $n$~vertices that contain only $1.33n$~edges \cite{auer2013onSparseMaximal}, while the maximum number of edges in saturated $2$-planar drawings is $5n$ \cite{pach1997graphs}. 
In particular, \cref{obs:lower_bound_2} shows that even if we fix the edge-coloring that certifies geometric thickness~$k$ (when considering precolored drawings), the number of edges in every saturated convex $\Theta^2$-drawing is $3n-6$.

\subsection{Edge-Density of Saturated Convex \texorpdfstring{$\bm{\Theta^3}$}{Theta-3}-Drawings}
\label{subsec:small_k}
With $k=2$ being covered by the general bounds of the previous section, we now turn to~$k=3$.
In the case of $\Theta^3$-drawings, we can strengthen the result of \cref{obs:face_complexity} as follows.

\begin{lemma}
\label{lem:convex_3-drawing_color_class_is_quadrangulation}
    If~$\Gamma$ is a saturated precolored convex $\Theta^3$-drawing, then each color class of~$\Gamma$ and the outer cycle forms an outerplane drawing~$\Gamma'$ where all inner faces have size at most~$4$.
\end{lemma}
\begin{proof}
    Let~$\Gamma$ be a saturated precolored convex $\Theta^3$-drawing with colors \emph{\blue}, \emph{\green} and \emph{\red} and let~$\Gamma'$ be the outerplane drawing induced by the \red edges and the outer cycle.
    Suppose some inner face~$f$ of~$\Gamma'$ contains at least five vertices~$v_1, \dots, v_5$.
    Each diagonal~$v_iv_j$ with $i\neq j$ is colored in \blue or \green.
    Note that the conflict graph~$H$ whose vertices are the diagonals~$v_iv_j$ and whose edges are pairs of crossing diagonals is a $5$-cycle, see \cref{fig:thickness_3_convex_free} for an example.
    \begin{figure}
        \centering
           \subcaptionbox{\label{fig:thickness_3_convex_free-1}}{\includegraphics[page=1]{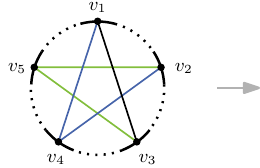}}
          \subcaptionbox{\label{fig:thickness_3_convex_free-2}}{\includegraphics[page=2]{figures/thickness_3_convex_free.pdf}}
        \caption{(\subref{fig:thickness_3_convex_free-1}) Five vertices on a face of size at least~$5$ in~$\Gamma_{\mathrm{r}}$ and their diagonals in~$\Gamma$. The black edge is not present ins~$\Gamma$. (\subref{fig:thickness_3_convex_free-2}) The corresponding vertex-coloring of the conflict graph~$H$.}
        \label{fig:thickness_3_convex_free}
    \end{figure}
    Yet, the $2$-edge-coloring of the diagonals induces a proper $2$-vertex coloring of the $5$-cycle~$H$, a contradiction.
    Thus, every inner face has size at most~$4$.
\end{proof}

Thus, every color class of a saturated convex precolored $\Theta^3$-drawing together with the outer cycle forms an outerplane drawing that contains an inner quadrangulation.
Now \cref{obs:edges_in_inner_triang_inner_quadrang} yields the following.
\begin{theorem}
    Every saturated precolored convex $\Theta^3$-drawing~$\Gamma$ of a graph~$G$ on $n \geq 3$ vertices contains at least $\frac{5}{2}n-6$ edges.
\end{theorem}

If a $\Theta^3$-drawing is saturated for every $3$-edge-coloring (with no monochromatic crossings), the lower bound on the number of edges can be improved.

\begin{theorem}
\label{obs:lower_bound_3}
    Every saturated convex~$\Theta^3$-drawing of a graph~$G$ on $n \geq 3$ vertices contains at least $\frac{7}{2}n-8$~edges.
\end{theorem}
\begin{proof}
    Let $\Gamma$ be a saturated convex~$\Theta^3$-drawing of $G$.
    That is, no edge can be added to~$\Gamma$, independent of the $3$-edge-coloring we consider.
    We call the three colors of a corresponding edge-coloring of~$\Gamma$ \emph{\blue}, \emph{\green}, and \emph{\red}. 
    Greedily adding missing diagonals in \blue or \green, we may assume that the union of the \blue edges, the \green edges, and the outer cycle is a saturated $\Theta^2$-drawing~$\Gamma'$. 
    In fact, as $\Gamma$ is saturated, we only recolor some \red edges in the process.
    By \cref{obs:lower_bound_2}, the subdrawing~$\Gamma'$ contains $3n-6$~edges. 

    It remains to show that there are at least $\frac{n}{2}-2$ \red inner edges.
    As $\Gamma$ is saturated (for every coloring), the \red edges together with the outer cycle form a drawing that contains an inner quadrangulation (cf. \cref{lem:convex_3-drawing_color_class_is_quadrangulation}).
    Thus, by \cref{obs:edges_in_inner_triang_inner_quadrang}, there are at least $\frac{n}{2}-2$ \red inner~edges.
\end{proof}

\section{Moving towards non-convexity in the free setting for \texorpdfstring{$\bm{k = 2}$}{k=2}}
\label{sec:non-convex-free}

In this section, we consider the more general case where the vertices of $G$ are not necessarily in convex position.
We show that, for $k = 2$, the lower bound from \Cref{sec:convex-free} (cf. \cref{obs:lower_bound_2}) extends to the general case, i.e., we prove the following theorem.

\begin{restatable}[\restateref{thm:non-convex-2-lower-bound}]{theorem}{nonConvex}
\label{thm:non-convex-2-lower-bound}
    Every saturated $\Theta^2$-drawing of a graph $G$ on $n \geq 3$ vertices contains at least $3n-6$ edges.
\end{restatable}
\begin{proof}[Proof Sketch]
Let $\Gamma$ be a $\Theta^2$-drawing of $G$.
We show that we can always add additional edges to $\Gamma$ without increasing its thickness to more than 2 if $\Gamma$ contains fewer than $3n-6$ edges.
We assume that the edges of $\Gamma$ are colored \blue and \red according to an arbitrary certificate of its thickness.
Let~$n'$ be the number of vertices that lie on the outer cycle. Adding missing edges and recoloring some of the \red edges in \blue, we greedily turn the \blue edges into a plane graph where each inner face is a triangle and the outer face is bounded by the outer cycle. That is, we may assume that there are $3n-6-(n'-3)$ \blue edges by \cref{obs:edges_in_inner_triang_inner_quadrang}.
In order to obtain the desired lower bound of $3n-6$ edges, we thus need to obtain at least $n'-3$ \red edges overall.

\begin{figure}
    \centering
    \subcaptionbox{\label{fig:nonConvex-1}}{\includegraphics[page=4]{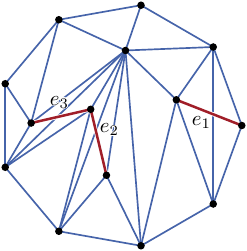}}
    \hfil
    \subcaptionbox{\label{fig:nonConvex-2}}{\includegraphics[page=5]{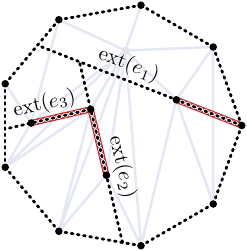}}
    \hfil
    \subcaptionbox{\label{fig:nonConvex-3}}{\includegraphics[page=6]{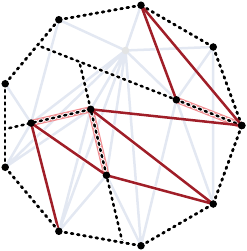}}
    \caption{(\subref{fig:nonConvex-1}) A $\Theta^2$-drawing $\Gamma$ where the blue edges form an inner triangulation. (\subref{fig:nonConvex-2}) The corresponding drawing $\Lambda$. (\subref{fig:nonConvex-3}) Triangulating each inner cell of $\Lambda$ yields seven additional red edges.
    Note that, while the red edges do not form an inner triangulation of the whole graph, we now have at least $3n-6$ edges overall as desired.}
    \label{fig:nonConvex}
\end{figure}

Consider an ordering $e_1, \dots, e_t$ of the \red edges.
We iteratively extend each edge~$e_i$ to a line segment as follows. 
We say that two line segments \emph{cross} if there exists a point $p$ that lies in the interior of both segments (i.e., $p$ is not an endpoint of either segment).
Let~$\ell_i$ be the supporting line of $e_i$. We define the \emph{edge extension of~$e_i$}, denoted~$\ext(e_i)$, as the segment of~$\ell_i$ of maximum length that contains~$e_i$ and does not cross any $e_j$ with $j \neq i$, the outer cycle, or any extension~$\ext(e_j)$ with $j < i$; see \Cref{fig:nonConvex}.
Note that, if two edge extensions~$\ext(e_i)$ and $\ext(e_j)$ share a point~$p$, then~$p$ is an endpoint of at least one of them.
We say that $\ext(e_i)$ and $\ext(e_j)$ \emph{touch} in the point~$p$.
If an extension $\ext(e_j)$ touches an extension~$\ext(e_i)$ in an inner point of~$\ext(e_i)$, we say $\ext(e_j)$ splits $\ext(e_i)$ into \emph{segments}.
Observe that every vertex that does not lie on the outer cycle lies in the interior of exactly one edge extension, but may be the endpoint of other edge extensions.

We denote by~$\Lambda$ the drawing induced by the outer cycle together with all edge extensions.
The drawing $\Lambda$ splits the plane into regions that we call \emph{cells}.
We denote by~$C(\Lambda)$ the set of inner cells of the drawing~$\Lambda$.
The \emph{boundary} of a cell~$c$ corresponds to all segments and vertices incident to~$c$. 
We let $\norm{c}$ denote the number of vertices on the boundary of $c$.

Each cell $c \in C(\Lambda)$ is convex.
Using a double counting argument for the vertex-cell incidences, we can show that the sum of these values over all cells of~$\Lambda$ plus the initial number of \red inner edges in $\Gamma$ adds up to at least $n' - 3$, the desired number of \red~edges.
\end{proof}

\section{Conclusion}\label{sec:conclusion}

We investigated saturated geometric drawings of graphs on $n$ vertices with geometric thickness~$k$. We provided upper and lower bounds on the number of edges in such drawings, and took a closer look at drawings of thickness $k=2$ and $k=3$. 
Several questions remain open, e.g., tight bounds for the convex case, and lower and upper bounds for min-saturated drawings with~$n'$ vertices on the convex hull.

\bibliography{refs}

\newpage

\appendix

\section{Full Proofs of \texorpdfstring{\cref{sec:convex-free}}{Section 2}}

\facecomplexity*
\label{obs:face_complexity*}
\begin{proof}
    Let $\Gamma$ be a (not necessarily saturated) precolored convex $\Theta^k$-drawing and let~$\Gamma'$ be the outerplane embedding induced by one color class (which we call \emph{\blue}) and the outer cycle. 
    Suppose there exists a face~$f$ of~$\Gamma'$ of size at least~$2k$.
    We need to show that~$\Gamma$ is not saturated.
    Note that there is a set of $k$~diagonals of $f$ which all pairwise intersect.
    Since no two such diagonals lie in the same color class and in particular no such diagonal is \blue, at most $k-1$~of the diagonals are part of the drawing~$\Gamma$.
    As we can add the missing diagonal in \blue to~$\Gamma$, the drawing~$\Gamma$ is not saturated.
\end{proof}

\section{Omitted Proofs of \texorpdfstring{\cref{sec:non-convex-free}}{Section 3}}

Because every cell $c \in C(\Lambda)$ is convex, the vertices on its boundary are in convex position and we can create a red inner triangulation in $c$ and, by \cref{obs:edges_in_inner_triang_inner_quadrang}, we obtain
$\max\{0, \norm{c} - 3\}$ \red inner edges between vertices on the boundary of $c$.
Therefore, our goal is to show that the sum of these values over all cells of~$\Lambda$ plus the initial number of \red inner edges in $\Gamma$ adds up to at least $n' - 3$.

Let $\Gamma_r$ be the subdrawing of $\Gamma$ induced by the \red edges.
In particular, note that $\Gamma_r$ thus only contains vertices that are incident to \red edges.
With the following two propositions, we show that 
the number of inner cells in $\Lambda$ that contain a vertex $v\in V(\Gamma_r)$ on their boundary as well as the total number of cells of $\Lambda$ can be directly obtained from~$\Gamma_r$.

\begin{lemma}
\label{prop:numIncidences}
    Every vertex $v \in V(\Gamma_r)$ lies on the boundary of $\deg_{\Gamma_r}(v) + 1$ inner cells of~$\Lambda$.
\end{lemma}
\begin{proof}
    Recall that we assume the vertices of $\Gamma$ are in general position and therefore $v$ only lies on edge extensions of edges incident to $v$.
    Also note that the outer edges are not contained in $\Gamma_r$.
    
    If $v$ lies on the convex hull of $\Gamma$, it is not contained in the interior of any segments and~$v$ is the endpoint of exactly $\deg_{\Gamma_r}(v)$ edge extensions.
    Since $v$ is also incident to two edges of the outer cycle, $v$ therefore lies on the boundary of exactly $\deg_{\Gamma_r}(v) + 1$ internal cells of $\Lambda$.

    If $v$ does not lie on the convex hull of $\Gamma$, recall that $v$ lies in the interior of exactly one edge extension, since our construction ensures that no two such extensions cross.
    Since every edge incident to $v$ induces an edge extension that has $v$ as its endpoint, $v$ lies on the boundary of $\deg_{\Gamma_r}(v) + 1$ cells, which concludes the proof.
\end{proof}

\begin{figure}
\centering
\subcaptionbox{\label{fig:planarization-1}}{\includegraphics[page=1]{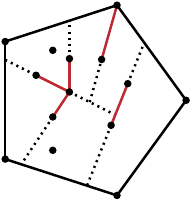}}
\hfil
\subcaptionbox{\label{fig:planarization-2}}{\includegraphics[page=1]{figures/Planarization_Sketch.pdf}}
\caption{Left: A drawing $\Lambda$ obtained after extending all inner \red edges. Right: its corresponding planarization $H$. Note that the cells of $\Lambda$ correspond bijectively to the faces of the planarization.}
\label{fig:planarization}
\end{figure}

\begin{lemma}
\label{prop:numFaces}
    $\abs{C(\Lambda)} = \abs{E(\Gamma_r)} + 1$.
\end{lemma}
\begin{proof}
    Let $x \coloneqq \abs{E(\Gamma_r)}$, i.e., $x$ denotes the number of edge extensions in $\Lambda$.
    Starting from~$\Lambda$, we first construct a planar graph $H$ and a corresponding planar drawing $\mathcal E$ by erasing all vertices of $\Lambda$ from the drawing and placing a vertex on every point where two edge extensions touch or an extension touches the outer cycle; see \Cref{fig:planarization} for an example.
    Note that vertices on the outer cycle that are not incident to red edges are not part of~$H$.
    In the context of this proof, let $n$, $m$, and $f$ denote the number of vertices, edges, and faces of $\mathcal E$, respectively.
    Observe that $f = \abs{C(\Lambda)} + 1$, since $f$ also accounts for the outer face of $\mathcal E$.
    
    We first show that $m = n + x$.
    Observe that every touching point (and thus the corresponding vertex of $H$) either lies on the convex hull or in the interior of exactly one edge extension (otherwise the segments would cross, a contradiction).
    Let $n_1$ and $n_2$ be the number of vertices satisfying the former and the latter property, respectively.
    Because the convex hull results in a cycle with $n_1$ vertices, 
    it contains $m_1 \coloneqq n_1$ edges.
    Note that the $x$ edge extensions of $\Lambda$ result in $x$ edge-disjoint paths in $\mathcal E$ that make up exactly the remaining $m_2 \coloneqq m - m_1$ edges of $\mathcal E$.
    Also note that each of the $n_2$ inner vertices of $\Lambda$ is an internal (i.e., non-endpoint) vertex of exactly one such path and that the number of edges of a path is the number of its internal vertices plus one.
    Thus, the paths contain $m_2 \coloneqq n_2 + x$ edges overall.
    Because $n = n_1 + n_2$, we obtain $m = m_1 + m_2 = n + x$.
     
    Note that $H$ is connected, since every vertex is incident to an edge of the convex hull or an edge extension of $\Lambda$.
    By Euler's Formula, we therefore obtain
    \[ \abs{C(\Lambda)} = f - 1 = (2 - n + m) - 1 = 1 - n + m = 1 - n + (n + x) = 1 + x = \abs{E(\Gamma_r)} + 1.\qedhere\]
\end{proof}

We can now show the required statement relating the number of red edges to $n'$.

\begin{lemma}
    $\displaystyle\abs{E(\Gamma_r)} + \sum_{c \in C(\Lambda)} (\norm{c} - 3) \geq n' - 3$.
\end{lemma}
\begin{proof}
    Let $V_c(\Gamma)$ denote the vertices of $\Gamma$ that lie on the outer cycle of $\Gamma$ but are not contained in $\Gamma_r$, i.e., that are not incident to \red edges.
    Moreover, we let $V_0(\Gamma) \coloneqq V(\Gamma) \setminus (V(\Gamma_r) \cup V_c(\Gamma))$ denote all remaining vertices, i.e., the vertices that lie on the inside of $\Gamma$ and are not incident to a \red edge.
    
    Note that it suffices to show the following:
    \[ \displaystyle\abs{E(\Gamma_r)} + \sum_{c \in C(\Lambda)} (\norm{c} - 3) = \abs{V(\Gamma)} - \abs{V_0(\Gamma)} - 3 \] 
    Since $n' \leq \abs{V(\Gamma)} - \abs{V_0(\Gamma)}$, the statement of the lemma then follows.
    
    Note that we can compute the combined size of all cells by summing up the number of vertex-cell incidences over all vertices. 
    Since the vertices of $V_c(\Gamma)$ are incident to exactly one cell in $\Lambda$ and vertices of $V_0(\Gamma)$ lie on the boundary of no cells, we obtain the following equation using \cref{prop:numIncidences}.
    \begin{equation}
        \label{eq:faces}
        \sum_{c \in C(\Lambda)} \norm{c} = \sum_{v \in V(\Gamma_r)} (\deg_{\Gamma_r}(v) + 1)+ \abs{V_c(\Gamma)}
    \end{equation}
    Moreover, the degree sum formula yields the following equation.
    \begin{equation}
        \label{eq:degSum}
        \sum_{v \in V(\Gamma_r)} \deg_{\Gamma_r}(v) = 2\abs{E(\Gamma_r)}
    \end{equation}
    Combining these equations with \Cref{prop:numFaces}, we obtain the desired statement as follows.

    \newcommand\oversetto[3][]{\overset{\text{\makebox*{#1}{#2}}}{#3}}
    \begin{align*}
        & \abs{E(\Gamma_r)} + \sum_{c \in C(\Lambda)} (\norm{c} - 3) \\
        &= \abs{E(\Gamma_r)} + \sum_{c \in C(\Lambda)} \norm{c} - 3|C(\Lambda)| & (\ref{eq:faces})\\
        &= \abs{E(\Gamma_r)} + \sum_{v \in V(\Gamma_r)} (\deg_{\Gamma_r}(v) + 1)+ \abs{V_c(\Gamma)} - 3|C(\Lambda)| & (\ref{eq:degSum})\\
        &= 3\abs{E(\Gamma_r)} + \abs{V(\Gamma_r)} + \abs{V_c(\Gamma)} - 3|C(\Lambda)|\\
        &= 3\abs{E(\Gamma_r)} + \abs{V(\Gamma)} - \abs{V_0(\Gamma)} - 3|C(\Lambda)|& \text{\cref{prop:numFaces}} \\
        &= \abs{V(\Gamma)} - \abs{V_0(\Gamma)}  - 3 & &\qedhere
    \end{align*}
\end{proof}

Since we now have $3n-6-(n'-3)$ \blue edges and at least $n'-3$ \red edges, we obtain the following result.

\nonConvex*
\label{thm:non-convex-2-lower-bound*}

\end{document}